%% file: osdp.tex
\documentclass[letterpaper, 10 pt, conference]{ieeeconf}  

\IEEEoverridecommandlockouts                              
\newtheorem{definition}{Definition}
\newtheorem{lemma}{Lemma}
\newtheorem{theorem}{Theorem}
\newtheorem{corollary}{Corollary}

\usepackage{times}
\usepackage{cite}
\usepackage{amsmath} 
\usepackage{amssymb}  
\usepackage{graphicx}
\title{Privacy Leakage over Dependent Attributes in One-Sided Differential Privacy}
\author{\thanks{This material is based on research sponsored by DARPA under agreement number FA8750-16-2-0021. The U.S. Government is authorized to reproduce and distribute reprints for Governmental purposes notwithstanding any copyright notation thereon. The views and conclusions contained herein are those of the authors and should not be interpreted as necessarily representing the official policies or endorsements, either expressed or implied, of DARPA or the U.S. Government.}}

\author{Phillip Lee, Kevin Smith
\thanks{P. Lee is with Honeywell.
{\tt{phillip.lee}@honeywell.com}}
\thanks{K. Smith is with Tridium.
{\tt{ksmith}@tridium.com}}}

\begin{document}
\maketitle

\begin{abstract}
Providing a provable privacy guarantees while maintaining the utility of data is a challenging task in many real-world applications. Recently, a new framework called One-Sided Differential Privacy (OSDP) was introduced that extends existing differential privacy approaches. OSDP increases the utility of the data by taking advantage of the fact that not all records are sensitive. However, the previous work assumed that all records are statistically independent from each other. Motivated by occupancy data in building management systems, this paper extends the existing one-sided differential privacy framework. In this paper, we quantify the overall privacy leakage when the adversary is given dependency information between the records. In addition, we show how an optimization problem can be constructed that efficiently trades off between the utility and privacy. 
\end{abstract}

\input{Introduction}
\input{Related}
\input{Model}

\input{Prob}

\input{main}
\input{Conclusion}

\bibliographystyle{IEEEtran}
\bibliography{osdp}

\end{document}

%% file: Introduction.tex
\section{Introduction}
\label{sec:intro}
Integration of communication and sensing capabilities into an ever-increasing array of physical systems enables intelligent automation services in many applications including energy, building management and transportation \cite{gubbi2013internet}. However, the ubiquitous nature of connected sensors also increases the risk of privacy breaches of individuals \cite{kozlov2012security}. Maximizing the utility of the sensed data while providing privacy guarantees to individuals has been an active area of research in both academia and industry \cite{mehrotra2016tippers,agarwal2010occupancy}.  

In particular, privacy-driven access control \cite{ni2010privacy} and differential privacy \cite{dwork2011differential} have been two promising areas of research. In general, the two approaches are complimentary to each other. Access control mechanisms regulate who can access potentially sensitive information according to pre-defined privacy rules. On the other hand, differential privacy is designed for a query-response model where randomized response algorithms are implemented in such a way that responses are statistically indistinguishable independent of the individual's presence in the data set.  

Differential privacy provides provable privacy guarantees based on the tunable parameter $\epsilon \geq 0$, where smaller $\epsilon$ implies more privacy-sensitive response. However, maintaining the utility of the released data even for a relatively large value of $\epsilon$ is a challenging task \cite{ghayyur2018iot}. The steep decline of utility even for a small increase of $\epsilon$ is largely due to the over-pessimistic assumption of differential privacy. In differential privacy, the randomized response is designed such that inclusion or exclusion of \emph{any} individual in the database should preserve statistical similarity of the response. In addition, every record in the database is treated as equally sensitive, resulting in decreased utility even when $\epsilon$ is relatively large. 

Recently, a new framework called one-sided differential privacy (OSDP) \cite{doudalis2017one} has been proposed to overcome this challenge. Unlike the existing differential privacy, OSDP assumes that not all records are equally sensitive and increases the utility of the released data by exploiting the release of non-sensitive records. 

In this paper, we study the privacy guarantees of OSDP in the context of building management system. Specifically, we examine privacy implications of occupancy data, use cases of occupancy data in the building management system, and the effectiveness of OSDP in this context. We generalize the privacy guarantees provided by OSDP to the case when records are statistically dependent, which would enable more informed parameter tuning of OSDP by taking into account the dependencies of records. We make the following specific contributions:
\begin{itemize}
\item We quantify the privacy leakage under OSDP when the records are statistically dependent. We generalize the previous privacy guarantee in \cite{doudalis2017one} against the exclusion attack when the adversary is assumed to have dependency information in the form of conditional probabilities between the records.
\item We present composition rules for the case of multiple queries as well as the case when two applications that utilize different attributes  are actively exchanging obtained query responses. In both cases, we show that multiplicative composition rules can be used to compute the privacy leakage. 
\item Using the results obtained from the previous contributions, we show that the total information leakage from the query responses under OSDP can be expressed explicitly as a function of the privacy parameters. Using information-theoretic metrics, we set up an optimization problem that efficiently trades off between utility and privacy. 
\end{itemize}

The paper is organized as follows. We review the related work in Section \ref{sec:related}. Section \ref{sec:model} provides a summary of one-sided differential privacy and the exclusion attack. The problem statement and motivation for this research are given in Section \ref{sec:problem}. Our main results are presented in Section \ref{sec:main}. Section \ref{sec:conclusion} concludes the paper. 

%% file: Related.tex
\section{Related Work}
\label{sec:related}
Analysis of potential privacy leakage and mitigation strategies have been active areas of research \cite{ukil2014iot, medaglia2010overview}. Access control mechanisms for private data have been studied in \cite{ni2010privacy, byun2005purpose}. In \cite{ni2010privacy}, the existing role-based access control model is extended to incorporate complex privacy policies including purposes and obligations while detecting potential conflicts in privacy policies. 
Similarly in \cite{byun2005purpose}, a privacy-preserving access control mechanisms for a relational database has been proposed by associating purpose information for each data element. 

Differential privacy \cite{dwork2011differential} has gained attraction from the research community as a promising framework that provides provable privacy guarantees. In differential privacy, the query response is randomized in such a way that any query response would be statistically similar in the presence or absence of presence of any individual in the database. In addition to the rigorous theoretical privacy guarantees, differential privacy allows the user to tune the privacy parameter $\epsilon$ to trade-off between utility and privacy.

Differential privacy has been applied to building applications in \cite{chen2017pegasus, mehrotra2016tippers, ghayyur2018iot} in the context of streaming occupancy data. In \cite{ghayyur2018iot}, utility of the occupancy data has been studied when varying values of $\epsilon$ have been applied to the data. While \cite{ghayyur2018iot} finds that information at the aggregated level are mainly preserved under differential privacy, preserving the utility of data at the individual level remains a challenge.

Recently, OSDP \cite{doudalis2017one} has been proposed to overcome this challenge. However, the privacy guarantee when the records in the database are statistically dependent remains an open problem as noted in \cite{doudalis2017one}.

%% file: Model.tex
\section{Model and Preliminaries}
\label{sec:model}
\subsection{One-Sided Differential Privacy (OSDP) and the Exclusion Attack}
One main drawback of differential privacy \cite{dwork2011differential} is that it treats every record to be sensitive. Such over-pessimistic assumption will lead to severe degradation of the utility of the released data. Often, not all records in the database would be sensitive. For simplicity, it is assumed that there exists a policy function ${P}$ which classifies each record as either sensitive ($P(r) = 0$) or non-sensitive ($P(r) = 1$). 

One approach is to not release a sensitive record as a query answer by returning nothing or rejecting the query \cite{rizvi2004extending} and by releasing only the non-sensitive records as responses to queries. 
However, such an approach will immediately result in the adversary inferring that the non-released record is sensitive, which results in privacy leakage for the owner of the non-released record. 

Such a privacy breach resulting from not releasing certain records is referred to as the \emph{exclusion attack} \cite{doudalis2017one}. One quantification of measuring the robustness of a query answering mechanism $\mathcal{M}$ against the exclusion attack is defined in \cite{doudalis2017one}. Range of a mechanism $\mathcal{M}$ is the set of all possible outcomes of query responses. 

\begin{definition} 
\label{def:exclusion}
($\epsilon$-Freedom from exclusion attacks) A mechanism $\mathcal{M}$ satisfies $\epsilon$-freedom from exclusion attacks for policy $P$ and and parameter $\epsilon$ if:
\begin{eqnarray*}
\label{eq:exclusion}
\forall x: P(x) = 0 \mbox{ and } \mathcal{O} \subseteq \mbox{range}(\mathcal{M})\\
\frac{\mathbb{P}(r = x| \mathcal{M}(D) \in \mathcal{O})}{\mathbb{P}(r = y| \mathcal{M}(D) \in \mathcal{O})} \leq e^{\epsilon} \frac{\mathbb{P}(r = x)}{\mathbb{P}(r = y)}
\end{eqnarray*}
where $r$ is the target record in the database $D$, $x$ is the value that makes the record sensitive, $y$ is another value in the domain. 
\end{definition}

Definition \ref{def:exclusion} states that even after observing the query answer, the ratio of the posterior probabilities of whether the target record was sensitive or some other value should remain \emph{similar} to that of the ratio of prior probabilities given to the adversary.

In this paper, we focus on one specific mechanism of achieving OSDP called One-Sided Differential Privacy Randomized Response (OSDPRR). 
OSDPRR is an algorithm to release true data \cite{doudalis2017one} while preserving privacy. The following definition describes the OSDPRR.
\begin{definition} 
\label{def:osdprr}
 (ODSPRR) For a record $r$, release the record with probability $1-e^{-\epsilon}$ where $\epsilon > 0$ if it is non-sensitive. If the record is sensitive, then do not release the record. 
\end{definition}

It is easy to see why OSDPRR would provide mitigation against the exclusion attack. The fact that a record has not been released does not automatically implicate that the record was sensitive since there is a non-zero probability of $e^{-\epsilon}$ that a non-sensitive record would not be released as well. 

In \cite{doudalis2017one}, it was proved that OSDPRR with parameter $\epsilon$ satisfies $\epsilon$-freedom from exclusion attacks under the assumption that \emph{records in the database are statistically independent from each other}. In this paper, we are interested in the case when this assumption does not hold, as often is the case in many applications.

%% file: Prob.tex
\section{Problem Statement and Motivation}
\label{sec:problem}
OSDP is an attractive solution for many practical applications due to its ease of implementation and provable privacy guarantees it provides. One application we studied is a building management system as part of the TIPPERS project \cite{mehrotra2016tippers}.

Figure \ref{fig:occupancy} shows an example of an occupancy pattern for a space in an office building generated by the Building Analytics App \cite{ghayyur2018iot}. Occupancy level 1 indicates that the particular space was occupied and 0 indicates otherwise. A quick inspection of the plot reveals that one can infer many attributes regarding the occupant of the space. For example, one can infer the starting time of the employee by examining when the occupancy level turns to 1 at the beginning of the day as well as when the employee leaves the office. In addition, one can also infer how many times the employee leaves the space and how long it takes for the employee to come back to the space on average. 

\begin{figure}[h]
    \centering
    \includegraphics[width=3.5in]{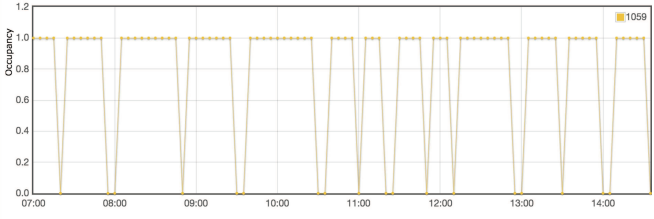}
    \caption{Occupancy pattern of an office space generated using Building Analytics App \cite{ghayyur2018iot}. Occupancy level 1 indicates that the space is occupied and level 0 indicates the space is not occupied.}
    \label{fig:occupancy}
\end{figure}

Such attributes can be useful in many building applications including energy as well as space optimization.
For example, a commercial building typically has a default work hours setting (ex: 8:00 am to 5:00 pm) when all lights are turned on. However, a lighting control coupled with occupancy sensors would enable more energy efficient lighting control \cite{thomas2009building}. This could be done by looking at average start time per space, and if the space does not become occupied after the default starting hour on average (ex: 9:00 am), then the lighting control can customize the lighting schedule to either match the average starting time of the space or only turn on the light based on the detected occupancy data.

Such lighting control will require average starting times of spaces occupancy to customize its schedule. However, the lighting control would not require the average occupancy level of spaces during the default work hours.  On the other hand, the average occupancy level (fraction of times when the space is occupied during the work hour) information would be required for other types of applications including space optimization \cite{mahasenan2018building}.

While the occupancy data and the attributes that pertain to the data can enable more efficient operation of buildings, many of these attributes are also potentially sensitive information for the occupants. One potential solution to provide privacy to the occupants while enabling these applications is OSDPRR. Given that different applications require different sets of attributes to perform their respective functionalities, OSDPRR can be implemented per attribute. However, as often is the case, there is no guarantee that the records from different attributes are statistically independent from each other. Moreover, it is also possible that two or more applications are colluding and actively exchanging information obtained from queries. 

\begin{figure}[h]
    \centering
    \includegraphics[width=3.6in]{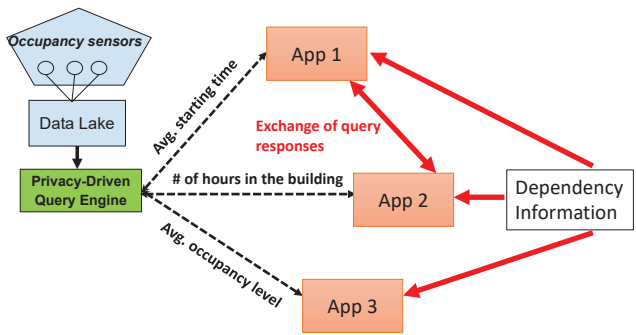}
    \caption{Illustration of an use case of occupancy data for building management. The dashed lines indicate queries and responses from multiple applications, and the red solid lines indicate additional information including the dependency information between the attributes as well as the exchanged queries responses between colluding applications. }
    \label{fig:framework}
\end{figure}

This implies that when releasing the data under OSDP, the query response mechanisms needs to take into account additional privacy leakage through statistical dependency as well as potential colluding scenarios. Figure \ref{fig:framework} illustrates an example of occupancy driven building applications where each application will be given the query response as well as additional information from statistical dependency of records and other query responses from exchanging information with other applications. 

In the following Section, we will show how to quantify privacy leakage from such additional information under OSDPRR.

%% file: main.tex
\section{Privacy Leakage over Dependent Attributes}
\label{sec:main}
In OSDPRR, there are only two possible outcomes: the record is either released or not. Throughout this paper, we denote 
\begin{displaymath}
    M_{i} = \left\{
    \begin{array}{cl}
    1, & \mbox{if $r_{i}$ is released.} \\
    0, & \mbox{if $r_{i}$ is suppressed.}
    \end{array}
    \right.
\end{displaymath}

In addition, we define a Bernoulli random variable to indicate whether a given record $r_{i}$ is sensitive or not. 
\begin{displaymath}
    X_{i} = \left\{
    \begin{array}{cl}
    1, & \mbox{if $r_{i}$ takes value that makes $r_{i}$ not sensitive.} \\
    0, & \mbox{else.}
    \end{array}
    \right.
\end{displaymath}

For the later results, we first quantify the robustness of the OSDPRR against the exclusion attack on record $r_{i}$ when the record is not released. 
\begin{theorem} The posterior probability ratio of $X_{i}$ given that $\mathcal{M}_{i} = 0$ is given as 
\begin{equation}
\label{eq:poterior}
\frac{\mathbb{P}(X_{i} =0 | \mathcal{M}_{i} = 0)}{\mathbb{P}(X_{i} =1 | \mathcal{M}_{i} = 0)} = e^{\epsilon_{i}} \frac{\mathbb{P}(X_{i} =0 )}{\mathbb{P}(X_{i} =1)} 
\end{equation}
Moreover, the posterior probability ratio of $X_{i}$ after $r_{i}$ is not released after $n$ independent consecutive queries is given as 
\begin{equation}
\label{eq:composition}
\frac{\mathbb{P}(X_{i} =0 | \mathcal{M}_{i}^{n} = 0)}{\mathbb{P}(X_{i} =1 | \mathcal{M}_{i}^{n} = 0)} = e^{n\epsilon_{i}} \frac{\mathbb{P}(X_{i} =0 )}{\mathbb{P}(X_{i} =1)} 
\end{equation}
where $\epsilon_{i}$ is the parameter chosen for record $r_{i}$ under OSDPRR. In other words, $r_{i}$ will be released with probability $1 -e^{\epsilon_{i}}$ if $r_{i}$ is not sensitive. 
\end{theorem}
\begin{proof} The proof can be found in \cite{doudalis2017one}.
\end{proof}

We are interested in how much information the adversary obtains regarding the sensitivity of $r_{j}$ drawn from another attribute when the statistically dependent record $r_{i}$ is not released or released. 

\subsection{One Application with Dependency Information}
\label{subsec:one}
First, we consider the case when one application obtains $\mathcal{M}_{i}$ and wants to infer the sensitivity of $X_{j}$ where record $r_{j}$ is used by another application. 
\begin{lemma}
\label{lemma1}    
$X_{j}$ is conditionally independent to $\mathcal{M}_{i}$ given $X_{i}$
\end{lemma}
\begin{proof} It is obvious that 
\begin{equation*}
\mathbb{P}(\mathcal{M}_{i}| X_{i}, X_{j}) = \mathbb{P}(\mathcal{M}_{i}| X_{i})
\end{equation*}
since $\mathcal{M}_{i}$ is a function that is determined only by $X_{i}$. Rewriting the above equation, we obtain
\begin{equation*}
    \frac{\mathbb{P}(\mathcal{M}_{i}, X_{i},X_{j})}{\mathbb{P}(X_{i},X_{j})} = \frac{\mathbb{P}(X_{j}|\mathcal{M}_{i},X_{i}) \mathbb{P}(\mathcal{M}_{i},X_{i})}{\mathbb{P}(X_{i},X_{j})} = \frac{\mathbb{P}(\mathcal{M}_{i},X_{i})}{\mathbb{P}(X_{i})}
\end{equation*}
    Rearranging the last equality and applying Bayes' theorem, we obtain
\begin{equation}
    \mathbb{P}(X_{j}|\mathcal{M}_{i},X_{i}) = \mathbb{P}(X_{j}|X_{i})
\end{equation}
\end{proof}
The intuition is that given $X_{i}$, there is no additional information that can be obtained from $\mathcal{M}_{i}$ regarding $X_{j}$. The following theorem quantifies the privacy leakage on $X_{j}$ when the record $r_{i}$ is not released. 

\begin{theorem} 
\label{thm:main}    
Let $\delta_{1}^{(ij)} = \mathbb{P}(X_{i} = 0| X_{j}=0)$ and $\delta_{2}^{(ij)} = \mathbb{P}(X_{i}=0|X_{j}=1)$, then the information leakage on $X_{j}$ from not releasing record $r_{i}$ can be quantified as 
    \begin{eqnarray}
    \frac{\mathbb{P}(X_{j}=0|\mathcal{M}_{i}=0)}{\mathbb{P}(X_{j}=1|\mathcal{M}_{i}=0)} = \frac{\delta_{1}^{(ij)}(e^{\epsilon_{i}}-1)+1}{\delta_{2}^{(ij)}(e^{\epsilon_{i}}-1)+1} \cdot \frac{\mathbb{P}(X_{j}=0)}{\mathbb{P}(X_{j}=1)}
    \end{eqnarray}
  where $\epsilon_{i}$ is the parameter for the OSDPRR for record $r_{i}$. 
    \end{theorem}
    \begin{proof} Posterior probability ratio of $X_{j}$ given $\mathcal{M}_{i} = 0$ can be written as the following.
    \begin{eqnarray}
    &&\frac{\mathbb{P}(X_{j}=0|\mathcal{M}_{i}=0)}{\mathbb{P}(X_{j}=1|\mathcal{M}_{i}=0)}\\
    &=& \frac{\sum_{x_{i}}\mathbb{P}(X_{i}=x_{i}|\mathcal{M}_{i}=0)\mathbb{P}(X_{j}=0|X_{i}=x_{i})}{\sum_{x_{i}}\mathbb{P}(X_{i}=x_{i}|\mathcal{M}_{i}=0)\mathbb{P}(X_{j}=1|X_{i}=x_{i})} \label{eq4}
    \end{eqnarray}
    This is because $\mathbb{P}(X_{j} = x_{j}| \mathcal{M}_{i}=0)$ can be decomposed as 
    \begin{eqnarray*}
    \mathbb{P}(X_{j} = x_{j}| \mathcal{M}_{i}=0) &=& \mathbb{P}(X_{j} = x_{j}, X_{i}=1| \mathcal{M}_{i}=0)\\
    &+& \mathbb{P}(X_{j} = x_{j}, X_{i}=0 | \mathcal{M}_{i}=0)
    \end{eqnarray*}
    and $\mathbb{P}(X_{j} = x_{j}, X_{i}=x_{i}| \mathcal{M}_{i}=0) = \mathbb{P}(X_{i} = x_{i}| \mathcal{M}_{i} = 0 ) \cdot \mathbb{P}(X_{j} = x_{j} | X_{i} = x_{i}, \mathcal{M}_{i} = 0)$. Moreover, from Lemma \ref{lemma1}, we have $\mathbb{P}(X_{j} = x_{j}| X_{i} = x_{i}, \mathcal{M}_{i} = 0) = \mathbb{P}(X_{j} = x_{j}| X_{i} = x_{i})$.
    
    Dividing both the numerator and denominator of expression (\ref{eq4}) by $\mathbb{P}(X_{i}=1|\mathcal{M}_{i}=0)$, and applying equation (\ref{eq:poterior}), we obtain
    \begin{eqnarray}
    &&\frac{\mathbb{P}(X_{j}=0|\mathcal{M}_{i}=0)}{\mathbb{P}(X_{j}=1|\mathcal{M}_{i}=0)}=\\
    && \frac{e^{\epsilon_{i}} \frac{\mathbb{P}(X_{i}=0)}{\mathbb{P}(X_{i}=1)}\mathbb{P}(X_{j}=0|X_{i}=0) + \mathbb{P}(X_{j}=0|X_{i}=1)}{e^{\epsilon_{i}} \frac{\mathbb{P}(X_{i}=0)}{\mathbb{P}(X_{i}=1)}\mathbb{P}(X_{j}=1|X_{i}=0) + \mathbb{P}(X_{j}=1|X_{i}=1)} \label{eq6}
    \end{eqnarray}
    However, the first terms in both numerator and denominator can be expressed as  
    \begin{eqnarray}
    && e^{\epsilon_{i}} \frac{\mathbb{P}(X_{i}=0)}{\mathbb{P}(X_{i}=1)}\cdot \mathbb{P}(X_{j}=x_{j}|X_{i}=0)\\ 
    &=&e^{\epsilon_{i}} \frac{\mathbb{P}(X_{i}=0)}{\mathbb{P}(X_{i}=1)}  \frac{\mathbb{P}(X_{i}=0|X_{j}=x_{j}) \mathbb{P}(X_{j}=x_{j})}{\mathbb{P}(X_{i}=0)}\\
    &=&e^{\epsilon_{i}} \frac{ \mathbb{P}(X_{i}=0|X_{j}=x_{j})}{\mathbb{P}(X_{i}=1)} \mathbb{P}(X_{j}=x_{j}) \label{eq9}
    \end{eqnarray}
    Replacing the first term of both the numerator and the denominator of (\ref{eq6}) with the expression (\ref{eq9}), we obtain the desired result.
    \end{proof}

    Here are some observations from Theorem \ref{thm:main}. If $\epsilon_{i} = 0$, then the record $r_{i}$ is not released even if it is not sensitive. This results in  $\frac{\delta_{1}^{(ij)}(e^{\epsilon_{i}}-1)+1}{\delta_{2}^{(ij)}(e^{\epsilon_{i}}-1)+1} = 1$. Therefore, no information is obtained regarding $X_{j}$, Also, if $X_{i}$ and $X_{j}$ are independent, then $\delta_{1}^{(ij)} = \delta_{2}^{(ij)}$, which yields the same result. On other other hand, as $\epsilon_{i} \to \infty$ ($r_{i}$ is released with probability 1 if it is not sensitive), then $\lim_{\epsilon_{i} \to \infty}\frac{\delta_{1}^{(ij)}(e^{\epsilon_{i}}-1)+1}{\delta_{2}^{(ij)}(e^{\epsilon_{i}}-1)+1} \to \frac{\delta_{1}^{(ij)}}{\delta_{2}^{(ij)}}$. Therefore, privacy guarantees regarding $X_{j}$ is completely determined by the statistical relationship between $X_{i}$ and $X_{j}$. 

    The following Corollary quantifies the information leakage of $X_{j}$ when $r_{i}$ is not released for $n$ consecutive queries under OSDPRR with parameter $\epsilon_{i}$.

    \begin{corollary}
    \label{cor:composition} Posterior probability ratio of $X_{j}$ given $\mathcal{M}_{i}^{n} = 0$ is given as 
    \begin{eqnarray}
        \frac{\mathbb{P}(X_{j}=0|\mathcal{M}_{i}^{n}=0)}{\mathbb{P}(X_{j}=1|\mathcal{M}_{i}^{n}=0)} = \frac{\delta_{1}^{(ij)}(e^{n\epsilon_{i}}-1)+1}{\delta_{2}^{(ij)}(e^{n\epsilon_{i}}-1)+1} \cdot \frac{\mathbb{P}(X_{j}=0)}{\mathbb{P}(X_{j}=1)}
    \end{eqnarray}
    where $\delta_{1}^{(ij)}$ and $\delta_{2}^{(ij)}$ are defined in Theorem \ref{thm:main}.
    \end{corollary}
    \begin{proof} The proof is straightforward from the proof of Theorem \ref{thm:main} by replacing $\mathcal{M}_{i}$ with $\mathcal{M}_{i}^{n}$ and by applying equation (\ref{eq:composition}) instead of equation (\ref{eq:poterior}). 
    \end{proof} 

    For completeness, we now derive the posterior probability ratio when $\mathcal{M}_{i} = 1$. 
    \begin{theorem}
    \label{thm:released}
    The posterior probability ratio of $X_{j}$ given that record $r_{i}$ is released is given as 
    \begin{equation}
        \frac{\mathbb{P}(X_{j}=0|\mathcal{M}_{i}=1)}{\mathbb{P}(X_{j}=1|\mathcal{M}_{i}=1)} = \frac{1-\delta_{1}^{(ij)}}{1-\delta_{2}^{(ij)}} \cdot  \frac{\mathbb{P}(X_{j}=0)}{\mathbb{P}(X_{j}=1)}
    \end{equation}
    \begin{proof} The posterior probability ratio of $X_{j}$ can be written as 
    \begin{eqnarray}
    && \frac{\mathbb{P}(X_{j}=0|\mathcal{M}_{i}=1)}{\mathbb{P}(X_{j}=1|\mathcal{M}_{i}=1)} \\
    &=& \frac{\mathbb{P}(X_{j} = 0, X_{i} = 1, \mathcal{M}_{i}=1)}{\mathbb{P}(X_{j} = 1, X_{i} = 1, \mathcal{M}_{i}=1)}\\
    &=& \frac{\mathbb{P}(X_{i}=1| X_{j} = 0)}{\mathbb{P}(X_{i}=1| X_{j} = 1)}\cdot \frac{\mathbb{P}(X_{j} = 0)}{\mathbb{P}(X_{j} = 1)}
    \end{eqnarray}
    where the first equality is from the fact that the probability of joint event of $X_{i} = 0$ and $\mathcal{M}_{i}=1$ is zero, and the second inequality is due to the fact that $\mathcal{M}_{i}$ is conditionally independent of $X_{j}$ given $X_{i}$. 
    \end{proof}

    \end{theorem}

    Theorems \ref{thm:main}, \ref{thm:released} and Corollary \ref{cor:composition} imply that the application that utilizes record $r_{i}$ can infer the likely sensitivity of $r_{j}$ from another attribute given additional information in the form of conditional probabilities between $X_{i}$ and $X_{j}$.

    \subsection{Colluding Applications with Dependency Information}
    \label{subsec:colluding}
    We now consider the case when colluding applications are actively exchanging information. In this case, an application will not only have $\mathcal{M}_{i}$ that it has obtained from a query but also $\mathcal{M}_{j}$ which is obtained from another application along with the dependency information.

    \begin{theorem}
    \label{thm:collusion}
    The posterior probability ratio of $X_{j}$ given $\mathcal{M}_{i} = 0$ and $\mathcal{M}_{j} = 0$ is given as  
    \begin{equation*}
     \frac{\mathbb{P}(X_{j}=0|\mathcal{M}_{i}=0, \mathcal{M}_{j} = 0)}{\mathbb{P}(X_{j}=1|\mathcal{M}_{i}=0, \mathcal{M}_{j} = 0)} = f_{1}^{(ij)}(\epsilon_{i}) e^{\epsilon_{j}} \frac{\mathbb{P}(X_{j}=0)}{\mathbb{P}(X_{j}=1)}
    \end{equation*}
    and 
    \begin{equation*}
        \frac{\mathbb{P}(X_{j}=0|\mathcal{M}_{i}=1, \mathcal{M}_{j} = 0)}{\mathbb{P}(X_{j}=1|\mathcal{M}_{i}=1, \mathcal{M}_{j} = 0)} = f_{2} e^{\epsilon_{j}} \frac{\mathbb{P}(X_{j}=0)}{\mathbb{P}(X_{j}=1)}
       \end{equation*}
    where $f_{1}^{(ij)}(\epsilon_{i})$ and $f_{2}^{(ij)}$ are defined as 
    \begin{equation}
    \label{eq:def}
    f_{1}^{(ij)}(\epsilon_{i}) = \frac{\delta_{1}^{(ij)}(e^{\epsilon_{i}}-1)+1}{\delta_{2}^{(ij)}(e^{\epsilon_{i}}-1)+1}, f_{2}^{(ij)} = \frac{1-\delta_1^{(ij)}}{1-\delta_2^{(ij)}}
    \end{equation}
    and $\epsilon_{i}$ and $\epsilon_{j}$ are parameters of OSDPRR for records $r_{i}$ and $r_{j}$ respectively. 
    \end{theorem}
    \begin{proof} In general, the posterior probability ratio can be written as 

    \begin{eqnarray}
    &&\frac{\mathbb{P}(X_{j}=0|\mathcal{M}_{i}=m_{i}, \mathcal{M}_{j} = m_{j})}{\mathbb{P}(X_{j}=1|\mathcal{M}_{i}=m_{i}, \mathcal{M}_{j} = m_{j})}\\
    &=&\frac{\mathbb{P}(X_{j} = 0, \mathcal{M}_{i} = m_{i}, \mathcal{M}_{j} = m_{j})}{\mathbb{P}(X_{j} = 1, \mathcal{M}_{i} = m_{i}, \mathcal{M}_{j} = m_{j})}\\
    &=& \frac{\mathbb{P}(X_{j} = 0|\mathcal{M}_{i} =m_{i})}{\mathbb{P}(X_{j} = 1 |\mathcal{M}_{i} =m_{i})} \cdot 
    \frac{\mathbb{P}(M_{j} = m_{j}| X_{j} = 0)}{\mathbb{P}(M_{j} = m_{j} | X_{j}=1)} \label{eq:general}
    \end{eqnarray}
    where the second equality is from the chain rule and the fact that $\mathcal{M}_{j}$ is conditionally independent to $\mathcal{M}_{i}$ given $X_{j}$. The first term of equation (\ref{eq:general}) is derived from Theorem \ref{thm:main} for $m_{i} =0$ and the second term is $e^{\epsilon_{j}}$ from the definition of OSDPRR for the case of $m_{j} = 0$. Similarly the case when $\mathcal{M}_{i} = 1$ is derived in Theorem \ref{thm:released}.
    \end{proof}

    Theorem \ref{thm:collusion} shows that similar to the case of multiple queries shown in (\ref{eq:composition}), the same multiplicative composition rule holds for the case of two colluding applications. 
    
    It is interesting to note that $f_{1}^{(ij)}(\epsilon_{i})$ or $f_{2}^{(ij)}$ defined in Theorem \ref{thm:collusion} is not necessarily greater than or equal to 1 unlike the multiplicative term $e^{\epsilon_{i}}$ in basic OSDP shown in (\ref{eq:poterior}). For the case when $\delta_{1}^{(ij)} < \delta_{2}^{(ij)}$, $f_{1}^{(ij)}(\epsilon_{i}) < 1$ for all $\epsilon_{i} \geq 0$. This implies that when given two pieces of information $\mathcal{M}_{i} = 0, \mathcal{M}_{j}=0$, it is possible to construct cases where the term $f_{1}^{(ij)}(\epsilon_{i}) e^{\epsilon_{j}}$ is approximately equal to 1. At a high level, $\mathcal{M}_{j} = 0$ increases the likelihood of $X_{j}$ being 0. At the same time, having $\mathcal{M}_{i}=0$ will increase the likelihood of $X_{i} =0$, but $X_{i}$ being 0 may increase the likelihood of $X_{j} =0$. In other words, the two pieces of information $\mathcal{M}_{i}$ and $\mathcal{M}_{j}$ may cancel each other's effect on the inference of $X_{j}$. Similarly, the same argument can be made when $\delta_{1}^{(ij)} < \delta_{2}^{(ij)}$ for the term $f_{2}^{(ij)} e^{\epsilon_{j}}$.  

    On the other hand, for the case when $\delta_{1}^{(ij)} > \delta_{2}^{(ij)}$, the two pieces of information $\mathcal{M}_{i} = \mathcal{M}_{j} = 0$ reinforces each other and result in higher likelihood of $X_{j} = 0$.

    \subsection{Choosing Privacy Parameters}
    \label{subsec:choice}
    Having higher values of $\epsilon_{i}$s will increase the utility of the data since non-sensitive records will be released with higher probabilities. However, when choosing a value of $\epsilon_{i}$, one needs to consider not only the information leakage of $X_{i}$ but also overall information leakage of $X_{j}$ through $\mathcal{M}_{i}$. 

    One established metric to quantify the average information leakage is mutual information \cite{cover2012elements}. We define the overall information leakage from $\mathcal{M}_{i}$ as $I_{i}$, mathematically defined as 
    \begin{equation}
    \label{def:information}
    I_{i}(\epsilon_{i}) = I(X_{i}; \mathcal{M}_{i}) + \sum_{j} I(X_{j}; \mathcal{M}_{i})
    \end{equation}
    In other words, $I_{i}$ is the sum of all mutual information between $X_{i}$, $X_{j}$ and $\mathcal{M}_{i}$. 
    The term $I_{i}(\epsilon_{i})$ can be computed using the results derived in the previous subsections \ref{subsec:one} and \ref{subsec:colluding}. The first term can be computed as 
    \begin{equation*}
    I(X_{i}; \mathcal{M}_{i}) = H(X_{i}) - \sum_{m_{i}}H(X_{i} | \mathcal{M}_{i} = m_{i}) \mathbb{P}(\mathcal{M}_{i} = m_{i})
    \end{equation*}
    where $H(X_{i})$ and $H(X_{i} | \mathcal{M}_{i} = m_{i})$ are defined as \cite{cover2012elements} 
    \begin{eqnarray*}
    &&H(X_{i}) = -\sum_{x_{i}} \mathbb{P}(X_{i} = x_{i}) \log \mathbb{P}(X_{i} = x_{i}),\\
    &&H(X_{i}| \mathcal{M}_{i} = m_{i})=\\
    &&-\sum_{x_{i}} \mathbb{P}(X_{i} = x_{i} | \mathcal{M}_{i} = m_{i}) \log \mathbb{P}(X_{i} = x_{i}|\mathcal{M}_{i} = m_{i})
    \end{eqnarray*}
    and $\mathbb{P}(\mathcal{M}_{i} = m_{i})$ can be computed as 
    \begin{eqnarray*}
    &&\mathbb{P}(\mathcal{M}_{i} = 0) = \\
    &&\mathbb{P}(\mathcal{M}_{i} = 0 | X_{i} =0) \mathbb{P}(X_{i} = 0)\\
    &+&\mathbb{P}(\mathcal{M}_{i} = 0 | X_{i} =1) \mathbb{P}(X_{i} = 1)\\ &=& \mathbb{P}(X_{i} = 0) + e^{-\epsilon_{i}} \mathbb{P}(X_{i} = 1)
    \end{eqnarray*}
    and similarly,
    \begin{eqnarray*}
    &&\mathbb{P}(\mathcal{M}_{i} = 1) =(1- e^{-\epsilon_{i}}) \mathbb{P}(X_{i} = 1)
    \end{eqnarray*}
    For simplicity, denote $\mathbb{P}(X_{i} = 0) = \theta_{i}$. Then, given these equations, we can simplify $I(X_{i}; \mathcal{M}_{i})$ as 
    \begin{equation*}
    I(X_{i}; \mathcal{M}_{i}) = H(X_{i}) - H_{2}(\theta_{1}^{(ii)}) \mathbb{P}(\mathcal{M}_{i} = 0)
    \end{equation*}
    where $\theta_{1}^{(ii)} = \frac{\theta_{i}}{e^{-\epsilon_{i}}(1-\theta_{i}) + \theta_{i}}$, and $H_{2}(\theta)$ is the binary entropy defined as 
    \begin{equation*}
    H_{2}(\theta) = -\theta \log\theta - (1-\theta) \log(1-\theta)
    \end{equation*}
    Similarly, $I(X_{j}; \mathcal{M}_{i})$ is given as 
    \begin{eqnarray*}
    I(X_{j}; \mathcal{M}_{i}) &=& H(X_{j}) - H_{2}(\theta_1^{(ij)})(\theta_{i} + e^{-\epsilon_{i}}(1-\theta_{i}))\\
    &-& H_{2}(\theta_2^{(ij)})(1-e^{-\epsilon_{i}})(1-\theta_i)
    \end{eqnarray*}
    where $\theta_{j} = \mathbb{P}(X_{j} = 0)$ and $\theta_{1}^{(ij)}$ and $\theta_{2}^{(ij)}$ can be computed from equations in Theorems \ref{thm:main} and \ref{thm:released} given as
    \begin{equation*}
    \frac{\theta_{1}^{(ij)}}{1-\theta_{1}^{(ij)}} = f_{1}^{(ij)}(\epsilon_{i}) \frac{\theta_{j}}{1-\theta_{j}}, \frac{\theta_{2}^{(ij)}}{1-\theta_{2}^{(ij)}} = f_{2}^{(ij)} \frac{\theta_{j}}{1-\theta_{j}}
    \end{equation*}
    Therefore, $I_{i}(\epsilon_{i})$ can be written explicitly as a function of $\epsilon_{i}$. Given this, one possible optimization problem that trades off between utility and privacy can be given as  
    \begin{eqnarray*}
    &&\mbox{max.} \sum_{i} \epsilon_{i} \\
    &&\mbox{subject to } \sum_{i}I_{i}(\epsilon_{i}) \leq T, \epsilon_{i}\geq 0
    \end{eqnarray*}
    where $T$ is some threshold value predefined by the query engine. The above optimization problem states that the sum of the $\epsilon_{i}$ values should be maximized subject to the constraint that the total information leakage is less than or equal to some threshold value. It should be noted that the optimization problem is not necessarily a convex optimization problem \cite{boyd2004convex}. This is because $f_{1}^{(ij)}(\epsilon_{i})$ is not necessarily a convex nor concave function for all values of $\epsilon_{i}$. However, as long as the dependency information do not rapidly change over time, the optimization problem can be solved offline once and the values obtained can be used for a prolonged period of time until the dependency information has significantly changed. Finding a convex relaxation that does provide provable optimality bound will be part of future work. 

%% file: Conclusion.tex
\section{Conclusions}
\label{sec:conclusion}
In this paper, we studied the privacy leakage in ODSP when the records are statistically dependent. We showed that the robustness against the exclusion attack can be quantified in a closed-form in the case of a single query as well as the case of multiple independent queries. 

We also considered the case when the applications are actively colluding with each other by exchanging query responses and quantified the privacy leakage. We showed that multiplicative composition rule can be derived from two pieces of information to quantify the overall privacy leakage.

Finally, we set up an optimization framework that trades-off between the utility and the privacy leakage of data.

While this paper studied the overall information leakage in terms of sensitivity of attributes denoted as $X_{i}$, extending the approach to describe the complete conditional distribution space of a dependent attribute is an interesting future work. In addition, we will investigate how the privacy leakage change when only partial dependency information is given to the adversary.